\documentclass[11pt]{amsart}
\allowdisplaybreaks[1]

\usepackage{a4wide}
\usepackage{amssymb}
\usepackage{amsmath}

\newtheorem{theorem}{Theorem}[section]
\newtheorem{lemma}[theorem]{Lemma}
\newtheorem{prop}[theorem]{Proposition}
\newtheorem{coro}[theorem]{Corollary}

\theoremstyle{definition}

\newtheorem{definition}{Definition}
\newtheorem{remark}{Remark}

\newcommand{\RR}{\mathbb{R}}

\newcommand{\CC}{\mathbb{C}}
\newcommand{\NN}{\mathbb{N}}
\newcommand{\TT}{\mathbb{T}}
\newcommand{\ZZ}{\mathbb{Z}}

\newcommand{\Bragg}{\mathcal{S}}
\newcommand{\Graph}{\mathcal{G}}
\newcommand{\ev}{\mathcal{E}}

\newcommand{\ext}{\mathcal{X}}

\newcommand{\cD}{\mathcal{D}}

\newcommand{\cZ}{\mathcal{Z}}

\newcommand{\vL}{\varLambda}

\newcommand{\dif}{\widehat{\gamma}}
\newcommand{\charF}{{\bf 1}}

\DeclareMathOperator{\vol}{vol}
\DeclareMathOperator{\len}{len}

\newcommand{\Hm}[1]{\leavevmode{\marginpar{\tiny%
$\hbox to 0mm{\hspace*{-0.5mm}$\leftarrow$\hss}%
\vcenter{\vrule depth 0.1mm height 0.1mm width \the\marginparwidth}%
\hbox to
0mm{\hss$\rightarrow$\hspace*{-0.5mm}}$\\\relax\raggedright #1}}}

\begin{document}

\title[Extinctions and Correlations ]{ Extinctions and Correlations \\for Uniformly Discrete Point Processes\\
with Pure Point Dynamical Spectra}

\author{Daniel Lenz}
\address{Mathematisches Institut, Fakult\"at f\"ur Mathematik und Informatik, Friedrich-Schiller-Universit\"at Jena,
D- 07737 Jena, Germany}
\email{daniel.lenz@uni-jena.de}
\urladdr{http://www.tu-chemnitz.de/mathematik/analysis/dlenz}

\author{Robert V.\ Moody}
\address{Department of Mathematics and Statistics,
University of Victoria, \newline
\hspace*{12pt}Victoria, BC, V8W3P4, Canada}
\email{rmoody@uvic.ca}
\date{\today}
\thanks{}

\begin{abstract}
  The paper investigates how correlations can completely specify a uniformly
  discrete point process. The setting is that of uniformly discrete   point sets in real space for which the corresponding dynamical
  hull is ergodic. The first result  is that all of the essential physical information
  in such a system is derivable from its $n$-point correlations, $n= 2, 3,
  \dots$.  If the system is pure point diffractive
an  upper bound on the number of correlations required can be derived
  from the cycle structure of a graph formed from the dynamical and Bragg
  spectra. In particular, if the diffraction has no extinctions, then the $2$ and $3$
  point correlations contain all the relevant information.
\end{abstract}

\maketitle

\section{The Setting}

\subsection{Quasicrystals and dynamical systems} The defining feature of physical cyrstals and quasicrystals is
the prominent appearance of Bragg peaks in their diffraction
diagrams. Mathematically the diffraction is a positive measure and
the Bragg peaks comprise the pure point component
of this measure. In a `perfect' crystal or quasicrystal, the diffraction should
be entirely pure point, and that is the situation that we shall assume here.
We shall simply refer to these as {\bf quasicrystals} in the sequel.

The diffraction does not, on its own, determine the internal structure of the
quasicrystal that created it. However, the diffraction does determine the
$2$-point correlation, which is its Fourier transform. Under appropriate
conditions (assumptions {\bf A(i),(ii)} below), knowledge of \emph{all}
the correlations ($2$-point, $3$-point, etc.) does determine the internal
structure (Thm.~ \ref{momentsToCorrelations}).  The primary objective of this paper is  to explore the details
behind this in the case that the diffraction is a pure point
measure.  In this case, under fairly mild conditions, one does not at all
need the entire set of correlations. In fact in the best situation, where there are no extinctions in
the Bragg spectrum (this term is explained below), the $2$- and $3$-point
correlations alone (in fact just the $3$-point correlations) are enough to determine the structure (Thm.~\ref{main2} and its Corollary).

Our approach here is to use a setting familiar from statistical mechanics and
from the theory of tilings and long-range aperiodic order. Rather than deal
with a single quasicrystal $\vL$, we work instead with translation invariant
families of them. The intuition is that such a family, $X$, will consist of
all those quasicrystals which are in some sense locally indistinguishable from
one another, or which cannot be isolated from one another by the physical
considerations at hand. As for the individual quasicrystals, we model these
simply as uniformly discrete\footnote{A subset $\vL$ of $d$-dimensional space
  $\RR^d$ is uniformly discrete (or more specifically $r$-uniformly discrete)
  if for some $r>0$ and for all $x,y \in \vL$ with $x \ne y$, $|x-y| \ge r$.}
point sets in space, with the points representing the positions of the atoms.
Since the dimension does not play any special role here, we work in general
$d$-dimensional space $\RR^d$.

Thus our main result deals with general ergodic
uniformly discrete point processes. It is always the case that the
dynamical spectrum is generated, as a group, by the diffraction spectrum. But here
we prove that if it has pure point spectrum and the dynamical
spectrum is expressible as a sum of {\em finitely many} copies of the diffraction spectrum then the point process is determined by its $n$-point-correlations
for some finite $n$. In the context of aperiodic order this becomes particularly relevant as it gives a way of assessing the degree of `complexity' of the long range order.

\subsection{Background}
It may be of interest to briefly discuss the reasons that such an elaborate formalism is relevant to what might seem a fairly straightforward exercise in spectral theory. Pure point diffraction from aperiodic structures was not predicted,
either by mathematicians or crystallographers. When it did appear,
both in tiling theory and experimentally in the discovery of aperiodic
metallic alloys, the projection method was quickly utilized and it was
generally believed that one could use standard techniques like the
Poisson summation formula (applied to lattices in higher dimensions)
to explain the diffraction.

It was A.~Hof who, in his much-cited papers \cite{Hof1,Hof2}, showed that
diffraction in aperiodic structures is not business as usual.
The Bragg spectrum of an aperiodic material is
not lattice-like as in the periodic case, but is typically dense in
Fourier space. The problem is that for a countable aperiodic set
$\Lambda$ of scatterers in $\mathbb R^d$ the Fourier transform $\hat
\nu$ of their Dirac comb $\nu = \sum_{x\in \Lambda} \delta_x$ is not
in general a measure. The sums involved diverge, even locally. This is
in contrast to the lattice case. For this reason the theory of
diffraction has developed by defining the diffraction as the Fourier
transformation of the volume averaged autocorrelation. It is the
`quadratic' nature of autocorrelation which  produces the necessary
convergence, and this is now the standard approach to diffraction in
the aperiodic case.

Assume now that we are in the case of a countable and uniformly
discrete set of scatterers. It was pointed out in
\cite{RW} that a good way to study an aperiodic set was to follow
ideas from statistical mechanics and form a compact space from its
translation orbit, the so-called hull $X$. This is a dynamical system
(with $\mathbb R^d$ as the acting group) and allows one to use
spectral theory. S. Dworkin \cite{Dworkin} then showed how the
dynamical spectrum could be linked to the diffraction spectrum by
using spectral measures. This linkage is now often called Dworkin's argument.

It was the use of hulls and Dworkin's argument that first allowed rigorous proof of the pure pointedness of model sets (or cut and project sets).
This is a good example of a situation where the result is seemingly clear from the Poisson summation formula, but on closer inspection one is confronted with divergent sums with no obvious mathematical meaning.

The first proof by A.~Hof, and its full generalization to all model
sets by M. Schlottmann \cite{Martin}, of the pure pointedness of
diffraction from model sets uses the fact that the hull is
measure-theoretically a compact Abelian group, and so pure point,
followed by Dworkin's argument using spectral measures. By the way,
unlike the periodic case, the hull is topologically {\em not} a group, and
indeed its highly subtle topology has been the focus of many
mathematicians recently (see \cite{KP,Bell} for reviews and further discussion).

Dworkin's argument still left the precise connection between the
diffraction and dynamical spectrum unresolved. Further developments on
the hull and its connection to diffraction and to point processes were
made in \cite{Bell} and \cite{Gouere} where it
is shown that under the assumption of ergodicity the autocorrelation
of the points sets of $X$ exist almost surely (in the sense of the
invariant measure $\mu$ on $X$) and almost surely are equal to the
first moment of the Palm measure of $X$.
In \cite{XR} this was extended to show that all
the higher correlations of $\Lambda \in X$ exist almost surely and
they completely determine $\mu$. This fact is not generally true for
point processes but here follows from the assumed uniform discreteness
of the point sets under consideration. In the pure point case it comes
pretty much for free from the spectral structure, as we see in the
present paper.

In \cite{LMS} (see \cite{Gouere,BL,LS} as well) it is proved that the diffraction is pure
point if and only if the dynamical spectrum is pure point. This is
remarkable since we know that in general the diffraction can fail to
see great chunks of the dynamical spectrum - even the pure point part
of the spectrum. In \cite{XR} the diffraction/dynamics connection is made
even more precise by showing that there is an isometric embedding of
$L^2(\mathbb R^d, \omega)$ into $L^2(X,\mu)$, where $\mathbb R^d$ is
the Fourier space with its diffraction measure $\omega$. This
embedding allows one to see how the eigenfunctions transfer across.
Precisely, each Bragg peak located at position $k$ gives rise to an
eigenfunction $f_k$ whose value at the point set $\Lambda \in X$ is
almost surely
\[ f_k(\Lambda) = \lim_{R\to\infty} \frac{1}{{\rm vol}\, C_R} \sum_{x \in \Lambda \cap C_R} e^{2 \pi i k\cdot x} \,.\]

We mention this, first because the convergence of this sum away from
$0$ is precisely what is called the Bombieri-Taylor conjecture in \cite{Hof1,Hof2}, (which was precisely that, until recently). There is a proof of
convergence in the $L^2$- sense in \cite{XR}. There is also a recent
proof of the point-wise convergence of the limit within the context  of uniform convergence in ergodic theorems in \cite{Lenz}. From the point
of view of our present paper, it is these functions $f_k$ which lie at
its heart. It is interesting to note that the way in which the isometric
embedding connection between diffraction and  dynamics is defined, the
eigenfunctions are nowhere in sight. It is only in the
$L^2$-completion that the eigenfunctions appear, and even then the way
in which they map (by the Bombieri-Taylor formula) is nothing like the
original defining map, and has to be proved.

Now we come to our present paper. Of course the underlying concern of
much of diffraction theory is that the inverse problem (resolving
structure from the diffraction) has no unique solution in general. The
problem is that the embedding of diffraction into dynamics is not
surjective. This problem is exacerbated in the aperiodic case.
Our setting is an ergodic dynamical
system of uniformly discrete point sets which are pure point
diffractive almost surely.  The main result of our paper is to show rather
precisely the significant role that extinctions (places in the
dynamical spectrum where there are no Bragg peaks) play in this
ambiguity.

In \cite{Mermin} D.~Mermin made the
remarkable suggestion that the second and third correlations should
always determine the structure. The argument made there is quite simple --
even trivial -- except that it again deals with quantities facing the
same problem originally pointed out in \cite{Hof1}; they are not convergent
(see also \S\ref{Origins}). The author was well aware of
the difficulties of his argument: he writes {\it ``Even
  granting that I have begged the question of when the density has a
  Fourier transform, and when the auto-correlation functions exist,
  this informal Fourier space argument that the identity of all second
  and third order correlations implies the identity of all higher
  order correlations is disarmingly trivial. I would very much like to
  learn of a comparably simple informal argument or an instructional
  counterexample in position space."}  Our paper is a response to this
question. If there are no extinctions, then the second and third correlations
do suffice. If the extinctions are not too bad, we can at least get away with knowing
only finitely many types of point correlations.  As further discussed in Section \ref{Final} 
there are also recent results which show that in fact there are situations where one really does need
higher moments than just the second and third to resolve aperiodic
point sets (in fact, even model sets), \cite{XR3}.

Apart from its mathematical interest and the potential directions for
further development, the results of our present paper seem to be
physically relevant. The detailed atomic structures of quasicrystals
are basically unknown in spite of over 20 years of work by
theoreticians and experimentalists. Model sets are one of the primary
modelling devices in the subject and diffraction is a fundamental
tool. It is relevant to know the controlling influences on diffraction
and to know how close diffraction, particularly in model sets, can
come to determining the underlying structure.

\subsection{Hulls}
The basic objects of interest in this paper are pairs $(X,\mu)$,
where $X$ is a set of $r$-uniformly discrete
point sets $\vL$ of some real space $\RR^d$ for some $r>0$
and $\mu$ is a probability measure on $X$. The assumptions that we need
to make on $(X,\mu)$ are listed in {\bf A(i), A(ii), A(iii)} below.

 \smallskip
Throughout the paper $C_R$, $R>0$, denotes the open cube $(-R/2,R/2)^d \subset
\RR^d$.

\smallskip
There is a uniform topology (called the {\bf local topology}) on the set $\cD_r(\RR^d)$ of all $r$-uniformly discrete subsets of $\RR^d$. The uniformity is generated from the collection
of all sets (entourages) of the form
$U(K,\epsilon)$, $K \subset \RR^d$ being compact, and $\epsilon >0$, where
\begin{equation}
U(K,\epsilon) = \{(\vL, \vL') \in \cD_r(\RR^d) \times \cD_r(\RR^d) \, : \,
\vL \cap K \subset \vL' + C_\epsilon \; \mbox{and} \;\vL' \cap K \subset \vL + C_\epsilon  \} \,.
\end{equation}

Thus $\vL$ and $\vL'$ are `close' if on some (`large') compact $K$
and for some (`small') $\epsilon >0$, the points of $\vL$ that are within $K$ also
lie in the $\epsilon$-cubical  neighbourhoods of the points of $\vL'$,
and vice-versa. It is relatively easy to see that $\cD_r(\RR^d)$ is compact
in this topology \cite{RW} and that the translation action $T$:
\[T_x(\vL) := x + \vL \in X, \;\mbox{for all}\;   x\in \RR^d\]
of $\RR^d$ on it is continuous. An alternative description of this topology using the functions $N_f$ below can be found in \cite{Bell}. 

We assume
\begin{itemize}
\item[\bf{A(i)}] $X$ is a closed translation invariant subset of $\cD_r(\RR^d)$;
\item[\bf{A(ii)}] $\mu$ is an ergodic probability Borel measure on $X$.
\end{itemize}

{\bf A(i)} obviously implies that $X$ is compact and, together with $\bf
A(ii)$, the pair $(X,\mu)$ along with the group action of $\RR^d$ by
translation is a dynamical system, both in the topological and measure
theoretic senses. The assumption of ergodicity is that $\mu$ is a translation
invariant measure and $X$ cannot be decomposed into two measurable invariant
subsets which each have positive measure. We can, if we wish interpret $\mu$
as a measure on $\cD_r(\RR^d)$ whose support lies inside $X$. This makes it clear that $\mu$ is the actual relevant piece of data. The space $X$ is only noted for convenience.

The basic open neighbourhoods of $\vL$ defined by the uniformity on $X$ are of the form
\[ U(K,\epsilon)[\vL] := \{ \vL' \in X \,:\, (\vL, \vL') \in U(K,\epsilon) \} \, . \]
These consist of the point sets $\vL'$ that are sufficiently close to making
the same pattern as $\vL$ within the compact set $K \subset \RR^d$. The
interpretation of $\mu(U(K,\epsilon)[\vL])$ is that it is the probability that
a random element $\vL'$ of $X$ will lie in $U(K,\epsilon)[\vL]$. The measure
$\mu$ thus gives the information about what patterns are possible and what
their probabilities of occurrence are. Its support specifies which subsets of
$\cD_r(\RR^d)$ are relevant. The ergodicity says that, when viewed from the
origin, the translations of any element $\vL$ from the support of $\mu$ will,
almost surely, faithfully represent all possible local patterns with the
correct frequencies.

We take the attitude that this is all we can hope to know about our physical system, and
thus our objective is to determine $\mu$ from other, physically observable, data. In our case this
other data will consist of various correlations of the system $(X,\mu)$.

\subsection{Diffraction and pure pointedness}

Let $S(\RR^d)$ denote the Schwartz space of all complex-valued infinitely many
times differentiable rapidly decreasing
functions on the real space $\RR^d$.

For each $n= 1,2, \dots$ the $n+1$-{\bf point correlation} of $\vL \in
\cD_r(\RR^d)$ is the measure $\gamma^{(n+1)}_\vL$ on $\RR^d \times \dots
\times \RR^d$ ($n$-factors) defined by
\[\gamma^{(n+1)}_\vL(F) = \lim_{R\to \infty} \frac{1}{\vol C_R} \sum_{x,y_1, \dots, y_n \in
  \vL \cap C_R} F(-x + y_1, \dots, -x+y_n) \, \] for all $F\in S(\RR^d \times
\dots \times \RR^d)$, if this limit exists.

\begin{theorem} \label{diffraction} \cite{Bell, Gouere, XR}
Let $(X,\mu)$ satisfy {\rm \bf A(i)} and {\rm \bf  A(ii)}. Then
\begin{itemize}
\item[(i)] for  $\mu$ almost every $\vL \in X$,  all of the $n$-point correlations $\gamma^{(n)}_\vL$ exist. Furthermore, they are almost surely
independent of the point-set $\vL$ chosen in $X$;
\item[(ii)] the $2$-point correlation is almost surely Fourier transformable.
\end{itemize}
\end{theorem}
The common $n$-point correlations are denoted simply as $\gamma^{(n)}$, and even
more simply as $\gamma$ for $n=2$. The Fourier transform of $\gamma^{(2)}_\vL$
is almost surely the Fourier transform $\dif$ of $\gamma$.

\begin{definition} $\dif$ is the {\bf diffraction} of $(X,\mu)$.
\end{definition}

Starting with the work of Hof \cite{Hof1} the  rigorous mathematical study  of diffraction for aperiodic order has attracted quite some attention in recent years. We refer to \cite{Lag,Lenz1,Lenz2} for recent surveys.

A basic idea is that the correlations are, in principle, quantities that can be
physically measured.  Certainly measurement of the diffraction is standard,
and hence its Fourier transform, the $2$-point correlation, may be considered
as known. There are reports of inference of higher correlations through
fluctuation microscopy \cite{TGFPN}, though it is not clear that these correlations go
beyond pair-pair correlations arising from squaring the $2$-point correlation.

\smallskip The space $L^2(X,\mu)$ of square integrable functions on $X$ gets
an $\RR^d$ action through translation of functions: $(T_t f)(x) = f(-t+x)$. A
simple consequence of the translation invariance of $\mu$ is that this action
is unitary for the basic inner product $\langle \cdot, \cdot \rangle$ defined by
\[ \langle f,g \rangle = \int_X f\overline g d \mu    \,. \]

As usual $f \in L^2(X,\mu)$ with $f\neq 0$ is an eigenfunction of $T$  (to the eigenvalue  $k\in \RR^d$)  if $T_t f = \exp(- 2 \pi i
  k\cdot t) f$ for all $t\in \RR^d$.  The fact that $\mu$ is assumed ergodic
  implies that the multiplicity of each eigenvalue is one.
The dynamical system $(X,\mu)$ is said to be {\bf pure point} if
$L^2(X,\mu)$ has a Hilbert basis of eigenfunctions.

A key point is a theorem that relates $L^2(\RR^d, \dif)$
and $L^2(X,\mu)$ and then relates the two concepts of pure pointedness.
In order to discuss this further we need some more notation.
For each $f\in S(\RR^d)$ let
$N_f: X \rightarrow \CC$ be defined by
$$N_f(\vL) = \sum_{x\in \vL} f(x).$$

\begin{lemma}\label{SW} The algebra generated by the $N_f$, $f\in S(\RR^d)$, is
  dense in the algebra of continuous functions on $X$ equiped with the
  supremum norm. In particular, it is dense in $L^2 (X,\mu)$.
\end{lemma}
\begin{proof} (See \cite{BL} for a similar argument.) Obviously, the
  algebra in question separates points, is closed under taking complex
  conjugates and to each $\varLambda\in X$, there exists an $f\in
  S(\RR^d)$ with $N_f (\varLambda)\neq 0$. Thus, the first statement
  follows by the Stone-Weierstrass theorem (see \cite{RS} for the
  version used here). The last statement is then clear.
\end{proof}

Define an action $U$ of $\RR^d$ on
$L^2(\RR^d, \dif)$ by
\[ (U_t f)(x) = e^{-2 \pi i t \cdot x}f(x)\]
for all $t,x \in \RR^d$, $f \in  L^2(\RR^d, \dif)$.

\begin{theorem}\label{basic}
Let $(X,\mu)$ satisfy {\rm \bf A(i)} and {\rm \bf  A(ii)}. Then the following holds.
\begin{itemize}
\item[(i)]  The set $\{\widehat{f} : f\in S(\RR^d)\}$ is dense in $L^2(\RR^d, \widehat\gamma)$ and  there is a unique isometric embedding
\[\theta : L^2(\RR^d, \widehat\gamma) \longrightarrow L^2(X,\mu)\, \]
with $\theta(\widehat{f}) = N_f $ for all $f\in S(\RR^d)$. This embedding
intertwines $U$ and $T$.
\item[(ii)] $\dif$ is a pure point measure if and only if $(X,\mu)$ is pure
  point.
\item[(iii)] For $k \in \RR^d$ the equation
 $$\theta(\charF_{k}) = \lim_{R\to \infty} \frac{1}{\vol C_R} \sum_{x\in \vL \cap C_R}
e^{2 \pi i k \cdot x}$$
holds, where the limit is meant in the $L^2$ sense. Moreover,
$ \dif(\{k\}) \ne 0$ if and only if $\theta(\charF_{k})\neq 0$. In this case,
$\theta(\charF_{k})$
is an eigenfunction of $(X,\mu)$ for the eigenvalue $k$.
\end{itemize}
\end{theorem}

\begin{remark} (a)  The theorem  may be found in the form stated here in
  \cite{XR}.  It has a long history that includes \cite{Dworkin, LMS, Gouere,
    BL, Lenz}. In fact, (i) and (iii) are shown in \cite{XR}, see
  \cite{Lenz} for extensions of (iii) as well. The statement (ii) is proven in various
  levels of generality in \cite{Dworkin, LMS,Gouere, BL, LS}.  We will give an independent proof below in Corollary \ref{pp}.

(b)  Under $\theta$ eigenfunctions go to eigenfunctions. However, the formula of
 part (i) is not applicable in (iii): the function $\charF_{\{ k\}}$, which
 takes the value $1$ at $k$ and zero everywhere else, is not even remotely in
 $S(\RR^d)$. The limit stated here appears only after approximation by
 functions from $S(\RR^d)$.

(c) As $\theta$ is an isometry,  $ \dif(\{k\}) \ne 0$ is obviously equivalent to  $\theta(\charF_{k})\neq 0$.
\end{remark}

The functions $\theta(\charF_{k})$ for $k$ with $\dif(\{k\}) \ne 0$ appearing
in (iii) of the previous theorem will play a crucial role in our
considerations. We define
$$ f_k := \theta(\charF_{k})\;\: \mbox{whenever} \;\: \dif(\{k\}) \ne 0.$$

\begin{definition} Let
\begin{eqnarray*}
 \ev &:=& \{k \in \RR^d \,:\, k \; \mbox{is an eigenvalue of} \; (X,\mu) \}\\
 \Bragg &:=& \{k \in \RR^d\,:\, \dif(\{k\}) \ne 0 \} \,.
 \end{eqnarray*}
 $\ev$ is the {\bf dynamical spectrum} of $(X,\mu)$ and $\Bragg$ is its {\bf Bragg spectrum}.
\end{definition}

We note that the Bragg spectrum is sometimes also known as Fourier-Bohr spectrum.  For convenience of notation, we will write $\dif(k)$ instead of $\dif (\{k\})$ in the remaining part of the paper. 

\smallskip

Let us briefly discuss the structure of $\ev$ and $\Bragg$ and their
relationship. It is well known that $\ev$ is a \emph{subgroup} of $\RR^d$. In fact, the product of
eigenfunctions is again an eigenfunction to the sum of the respective
eigenvalues and the complex conjugate of an eigenfunction is an eigenfunction
to the inverse of the corresponding eigenvalue.

As shown in  (iii) of Thm.~ \ref{basic}  any Bragg peak $k$ comes with a
canoncial eigenfunction $f_k   = \theta(\charF_k)$.  In particular,
we have the inclusion
\[ \Bragg \subset \ev. \]
In general, this inclusion is strict  even in the pure point case. The
limit formula for $f_k$ in Thm.~ \ref{basic} shows that
\begin{equation}\label{reflection}
 f_{-k} = \overline{f_k} \,
 \end{equation}
 and that
$$ f_0 \neq 0.$$
Thus, the statement on the eigenfunctions in (iii) of Thm.~ \ref{basic}
shows that $\Bragg$ satisfies
$$ 0\in \Bragg ,\:\;\mbox{and}\;\: \Bragg = - \Bragg.$$
It is a
fundamental fact that the canonical eigenfunction $f_k$ of the Bragg peak $k$
can be related to the intensity of the Bragg peak. More precisely, note that
 due to the ergodicity the modulus of  any eigenfunction is constant
 $\mu$-almost everywhere. Thus,
\begin{equation} \label{eigenLength}
|f_k| =  \langle  f_k,f_k \rangle_{L^2(X,\mu)}^{1/2} =
\langle  \charF_k,\charF_k \rangle_{L^2(\RR^d,\dif)} ^{1/2} = \dif(k)^{1/2} \,,
\end{equation}
where the first equality holds $\mu$-almost everywhere.

\smallskip

Our final assumption is:

\begin{itemize}
\item[\bf{A(iii)}] $\widehat \gamma$ is a pure point measure.
\end{itemize}

\begin{definition} A pair $(X,\mu)$ satisfying axioms {\bf A(i),(ii),(iii)}
is called a {\bf pure point ergodic uniformly discrete point process}.
\end{definition}

In this pure point case, $\Bragg$ generates $\ev$ as a group, as shown in
\cite{BL}. We will give an independent proof based on Thm.~ \ref{basic} in
Corollary \ref{pp} below.

\begin{definition} The set $\ext := \ev \backslash \Bragg$
is called the set of {\bf extinctions} of $(X,\mu)$.
\end{definition}

\subsection{Correlations and moments}
Let $(X,\mu)$ satisfy {\bf A(i)} and {\bf A(ii)}. For $n=1,2, \dots$, the
$n$th-{\bf moment} of $\mu$ is the measure $\mu_n$ defined on $\RR^d \times
\dots \times \RR^d$ ($n$-factors) by \[\mu_n(h_1, \dots , h_n) = \int_X N_{h_1}
\dots N_{h_n} \, d\mu \, \] for all $h_1, \dots, h_n \in S(\RR^d)$. It is
clear that these moment measures are invariant under simultaneous translation
of all the variables. There is a standard procedure of eliminating this
translation invariance resulting in the \emph{reduced} moments $\mu_n^{{\rm
    red}}$ which are in one less variable. This works as follows: let $g,h_1,
\dots, h_{n-1} \in S(\RR^d)$ and let $G$ be the function on $(\RR^d)^n$ whose
value on $(x,y_1, \dots, y_{n-1})$ is $g(x)(T_xh_1)(y_1) \cdots (T_x h_{n-1})(y_n)$.
Then
\[\mu_n(G) = \mu_n(g (T_xh_1) \dots (T_xh_{n-1}))
= \int_{\RR^d}g(x) dx \, \mu^{\rm red}_{n}(h_1,\dots, h_{n-1}) \, , \]
and this equation defines the reduced moments.

Most importantly for our purposes,
these reduced moments are also the correlations -- see \cite{DV-J}, Sec.~12.2  and
\cite{XR}, Sec.~7 \footnote{The reduced moments are also connected directly to Palm measures, a direction that is more fully explored in  \cite{Gouere, XR}.}. More precisely, the following holds.

\begin{theorem} \label{momentsToCorrelations}
Let $(X,\mu)$ satisfy {\bf A(i)}
and {\bf A(ii)}.
\begin{itemize}
\item[(i)] For each $m\in \NN$, $\mu$ is uniquely determined by its moments
  $\mu_n$, $n\geq m$.
\item[(ii)] $\gamma^{(n)} = \mu_n^{{\rm red}}$, $n= 2,3, \dots$.
\item[(iii)] For $n\geq 2$, $\mu_n$ is uniquely determined by  $\mu_n^{{\rm red}}$.
\item[(iv)] The measure $\mu$ is uniquely determined by $\gamma^{(n)}$, $n= 2,3, \dots$.
\end{itemize}
\end{theorem}
\begin{proof} (i) For $m=1$ the statement follows immediately from Lemma
  \ref{SW}.  Now, it suffices to show that the $\mu_n$, $n> m$, determine
  $\mu_m$. By Lemma \ref{SW}, again, the constant function
  $1$ can be approximated by elements of the algebra generated by the $N_h$,
  $h\in S(\RR^d)$. Thus, a
  product
$$ N_{h_1} \cdots N_{h_m} =  N_{h_1} \dots N_{h_m} . 1$$
can be approximated by  linear combinations of products of more than $m$
functions in $S(\RR^d)$. Thus, $\mu_n$, $n>m$, determine $\mu_m$.

The proof of (ii) can be found in \cite{DV-J}, Prop.~12.2.V.
The proof of (iii) can be found in \cite{DV-J}, Sec.~10.4.
Finally, (iv) is a direct consequence of (i), (ii) and (iii).
\end{proof}

\begin{remark} The proof of (i) in the previous theorem does not require
  ergodicity. It only uses that the functions $N_f$, $f\in S (\RR^d)$, are
  bounded and continuous on the compact $X$.
\end{remark}

The point of the previous theorem  is that rather than correlations, we may instead look at
corresponding moments. Our question becomes that of asking how many moment
measures are required to pin down $\mu$ uniquely.

\section{Eigenfunctions and cycles}

\subsection{The cycle function of $(X,\mu)$} \label{cyclefunction}
Let $(X,\mu)$ be a pure point ergodic uniformly discrete point process. Thus,
{\bf A(i),(ii),(iii)} are valid.

\bigskip

The elements of $L^2(\RR^d , \dif)$ are all the sums
\[ \sum_{k\in \Bragg} x_k \charF_{k}, \; \mbox{where} \; \sum_{k\in \Bragg}
|x_k|^2
\dif(k) < \infty \, .\]
As  described in (i) and (iii)  of Thm.~ \ref{basic} the map $\theta$ exhibits very different behaviour on functions
$h \in S(\RR^d)$ and on functions $\charF_k$. This leads to two very different ways
in which to write $\overline{\theta(L^2(\RR^d, \dif))}$.  More precisely, both
the linear span of the set of $ \charF_{k}$, $k\in \Bragg$, and the set $\widehat{h}$, $h\in
S(\RR^d)$, are dense in $L^2(\RR^d, \dif)$. As $\theta$ is an isometry, this gives
\[ \theta(L^2(\RR^d, \dif)) = \overline{\{N_h \,:\, h \in S(\RR^d)
  \}} =  \overline{{\mbox{linear span}}\, \{ f_k \,:\, k \in \Bragg \}} \, .\]
Our next aim is to obtain similar statements for products. This requires some care as we will have to deal with products of infinite sums. The corresponding details are given in the next two lemmas.

We will use repeatedly the elementary fact that
$ \{g_m h\}$ converges to $g h$ in $L^2$ whenever $\{g_m\}$ is a sequence  converging to $g$ in $L^2
$ and $h$ is a bounded function.

\begin{lemma}\label{formulan} Let $n\in \NN$  and  $h_1,\ldots,h_n \in  S(\RR^d)$ be given. Then,
  $$N_{ h_1} \dots N_{h_n} = \sum_{k_1\in
    \Bragg}\ldots\sum_{k_n\in \Bragg} \widehat{h_1} (k_1)\dots \widehat{h_n}
  (k_n) f_{k_1}\dots f_{k_n},$$
where the sums exist in $L^2$ and are taken one after the other. In particular,
$$
  \int N_{ h_1} \dots N_{ h_n} d\mu = \sum_{k_1\in
    \Bragg}\ldots\sum_{k_n\in \Bragg} \widehat{h_1} (k_1)\dots \widehat{h_n}
  (k_n) \mu (f_{k_1}\dots f_{k_n}).$$
\end{lemma}
\begin{proof} By (i) of Thm.~ \ref{basic} we have
$$ N_{h_j} = \sum_{k_j\in \Bragg} \widehat{h_j} (k_j) f_{k_j}$$
for each $j$.  Therefore,
$$ N_{ h_1} \dots N_{h_n} = \left( \sum_{k_1\in \Bragg} \widehat{h_1} (k_1) f_{k_1} \right) N_{h_2} \dots N_{h_n} = \sum_{k_1\in \Bragg} \widehat{h_1} (k_1) f_{k_1} N_{h_2} \dots N_{h_n}$$
and the first statement follows by induction.  As $\mu$ is a finite measure, the last statement then follows easily.
\end{proof}

As a corollary  of this lemma we obtain a new proof of the following known fact.

\begin{coro} \label{pp} $T$ has pure point spectrum, i.e. there exists a basis of $L^2 (X,\mu)$ consisting of eigenfunctions.  Moreover, any eigenfunction is a finite  product of  functions $f_k$, $k\in \Bragg$, and any eigenvalue is a sum of $k\in \Bragg$.
\end{coro}
\begin{proof}  By the previous lemma, any function of the form $N_{ h_1} \dots N_{h_n}$ with $h_j\in S(\RR^d)$ can be approximated by linear combinations of products  of the form $f_{k_1}\dots f_{k_n}$, $k_j\in \Bragg$.   Lemma \ref{SW} then gives  that functions of the form $f_{k_1}\dots f_{k_n}$, $n\in \NN$, $k_j\in \Bragg$, $j=1,\ldots,n$ are total in $L^2 (X,\mu)$. As each function of the form $f_{k_1}\dots f_{k_n}$ is an eigenfunction to the eigenvalue $k_1+\cdots + k_n$ the statement follows.
\end{proof}

\begin{lemma}\label{formulaf} Let $n\in \NN$ and  $k_1,\ldots,k_n\in \Bragg$ be given. For
  each $j=1,\ldots,n$, let $\{h_j^{(m)}\}$ be a sequence  in $S(\RR^d)$ whose
  Fourier transforms converge to  $\charF_{k_j}$ in $L^2 (\RR^d,\dif)$.
Then,  $$f_{k_1}\dots f_{k_n} = \lim_{m_1\to \infty}  \lim_{m_2\to \infty}
  \ldots \lim_{m_n\to \infty} N_{h_1^{(m_1)}} \dots N_{h_n^{(m_n)}},$$
where the limits are taken in $L^2$.
In particular,  $$\mu (f_{k_1}\dots f_{k_n}) = \lim_{m_1\to \infty}
  \lim_{m_2\to \infty} \ldots \lim_{m_n\to \infty} \mu (    N_{h_1^{(m_1)}} \dots N_{h_n^{(m_n)}})  \,.$$
\end{lemma}
\begin{proof}  The functions $N_{h_j^{(m)}}$ and the functions
  $f_{k_j}$, $j=1,\ldots,n$, $m\in \NN$ are bounded. Thus, the  convergence of the $ h_j^{(m)}$ easily yields  convergence of the products.  (Note that the
  limits are taken one after the other). As $\mu$ is a finite measure the last statement
  then follows easily.
\end{proof}

We will now introduce a crucial object in our studies, namely the cycle
function $a$.  Notice that, by almost sure constancy of the modulus of the functions  $f_k$,
\begin{eqnarray*}
 |f_{k_1} \dots f_{k_n}|^2 &=& \int_X f_{k_1} \dots f_{k_n}\overline{f_{k_n}} \dots
\overline{f_{k_1}} d\mu \\
&=& \langle f_{k_1}, f_{k_1}\rangle \dots  \langle f_{k_n}, f_{k_n}\rangle
= \dif(k_1) \dots \dif(k_n)\,.
\end{eqnarray*}
If $k_1 + \cdots +k_n = 0$ then $f_{k_1} \dots f_{k_n}$ is an eigenvector for
$0$, and hence is a multiple of the constant function $1_X$. Thus, in this
case,
\begin{equation} \label{defa}
f_{k_1} \dots f_{k_n} = a(k_1, \dots, k_n) \dif(k_1)^ {1/2} \dots \dif(k_n)^{1/2} \, 1_X \end{equation}
for some
\[a(k_1, \dots, k_n) \in U(1) \, .\]
Here, $U(1)$ is the unit circle i.e. the set of all complex numbers of modulus one. 
For any $k_1, \dots, k_n \in \Bragg$, we therefore obtain
\begin{eqnarray}\label{evaluationofmu}
\mu(f_{k_1} \dots f_{k_n}) &=& \int_X f_{k_1} \dots f_{k_n} \overline{1_X}\,
d\mu = \langle f_{k_1} \dots f_{k_n}, 1_X \rangle\\ \nonumber
&=& \begin{cases} a(k_1, \dots, k_n) \dif(k_1)^{1/2} \dots \dif(k_n)^{1/2} &\mbox{if} \; k_1 + \cdots +k_n = 0;\\
0 &\mbox{if} \; k_1 + \cdots +k_n \ne 0 \,,
\end{cases}
\end{eqnarray}
since eigenfunctions for different eigenvalues are orthogonal.

The next two results basically say that knowledge of the cycle function  $a$ determines the
moments and vice versa.

\begin{prop}\label{a2moments} Let $n\in \NN$ be given. Then, the  $n$-th moment of $\mu$  is
   uniquely determined by  $\dif$ and the quantities
\[a(k_1, \dots, k_n)\;\:\mbox{for} \; \:k_1, \dots , k_n \in \Bragg
\;\:\mbox{with}\: \; k_1 + \cdots + k_n = 0\,.\]
\end{prop}
\begin{proof} This follows directly from Lemma \ref{formulan} and
  \eqref{evaluationofmu}.
\end{proof}

\begin{prop}\label{moments2a} Let $n\in \NN$ be given. Then, the values
  $a(k_1,\ldots,k_n)$ for $k_1,\dots,k_n\in \Bragg$  with $k_1+ \dots + k_n =0$ are uniquely determined
  by $\dif$ and the $n$-th moment of $\mu$.
\end{prop}
\begin{proof} This follows from Lemma \ref{formulaf} and
  \eqref{evaluationofmu}.
\end{proof}

It is convenient to introduce the Cayley graph $\Graph$ of $\ev$ with respect
to the set of generators $\Bragg$. Its vertices are the points
of $\ev$ and its edges are the pairs $\{k,l\}$ of vertices whose differences
$k-l$ lie in $\Bragg$. Since $\Bragg= -\Bragg$, we may treat the edges as undirected.
Any ${\bf k} = (k_1, \dots k_n) \in
\Bragg^n$ with $k_1+\ldots + k_n =0$ leads to a cycle $\{0, k_1, k_1+k_2,
\dots, k_1 + \dots + k_{n-1}, k_1 + \dots + k_n = 0\}$ in $\Graph$.  Thus the
function $a$ described above can be thought of as a function of the set $Z$ of
cycles of $\Graph$. We shall call it the {\bf cycle function} of $(X,\mu)$.

\subsection{Properties of the cycle function}

Given
${\bf k} , {\bf l} \in Z$, their concatenation
\[{\bf kl} :=  (k_1, \dots, k_n,l_1,\dots, l_p) \,,\] is obviously also in $Z$.

\begin{prop} The cycle function $a$ has the following properties:
\begin{itemize}
\item[(i)] for all ${\bf k},{\bf l} \in Z$, $a({\bf k})a({\bf l}) = a({\bf k}{\bf l})$;
\item[(ii)] $a(0) = 1$;
\item[(iii)]  for all ${\bf k}\in Z$, $a({\bf k}, -{\bf k}) = 1$;
\item[(iv)] $a(k_1, \dots, k_n)$ is independent of the order of the
elements $k_1, \dots, k_n$ making up the cycle;
\item[(v)] given any cycle ${\bf k} \in Z$, then any pair $\{k, -k\}$ where $k \in \Bragg$,
and also $0$
can be inserted into or deleted from the symbols of ${\bf k}$ without affecting
the value of the cycle function $a$.
\end{itemize}
\end{prop}
\begin{proof} (i) follows from \eqref{defa}.  As for (ii), note that
$f_0 = a(0)\dif(0)^{1/2}$ by \eqref{defa}. Since $f_0 = \theta(\charF_{\{0\}}) \ge 0$ from
Thm.~\ref{basic}~(iii), $\dif(0) > 0$, and $a(0) \in U(1)$, we see that $a(0) =1$.

This proves (ii).  Part (iii) follows from

\[ a(k,-k) \dif(k) = f_k f_{-k} = f_k \overline{f_{k}} = |f_k|^2 = \dif(k). \]
 Items (iv) and (v) are trivial consequences of \eqref{defa} and
parts (i),(ii), and (iii).
\end{proof}

Let $Z_n:=\{(k_1,\dots,k_n) \in Z \}$, $Z_0:= \{\emptyset\}$,
and $Z_\infty = \bigcup_{n=0}^\infty Z_n$.
We introduce an equivalence relation on $Z_\infty$ by transitive extension
of the two rules:
\begin{itemize}
\item ${\bf k} \sim {\bf l}$ if ${\bf l}$ is a permutation of the symbols of ${\bf k}$
\item ${\bf k} \sim {\bf l}$ if ${\bf l}$ can be obtained from ${\bf k}$ by inserting
or removing pairs $\{k,-k\}$, $k\in \Bragg$, or by inserting or removing $0$.
Let
\[\cZ:=Z_\infty /\sim \,.\]
\end{itemize}
It is easy to see that ${\bf k} \sim {\bf l}\,,\, {\bf k'} \sim {\bf l'} \Rightarrow
{\bf k}{\bf l} \sim {\bf k'}{\bf l'}$ , so multiplication descends from $Z_\infty$ to
$\cZ$. Indeed $\cZ$ is an abelian group under this multiplication, with
$\emptyset^\sim$ as the identity element.

Of the various $(k_1, \dots, k_n)$ that can represent a given element
$\kappa \in \cZ$ there is (at least) one of minimal length $n$. This minimal
length is denoted by $\len(\kappa)$. Define
\[\cZ_n:= \{ \kappa \in \cZ \,:\, \len(\kappa) \leq n \}\,. \]
We shall also write $\len({\bf k}) = \len(\kappa)$ when ${\bf k}^\sim = \kappa$
and call it the {\bf reduced length} of ${\bf k}$.

\begin{eqnarray*}
\cZ &=& \bigcup_{n=0}^\infty \cZ_n,\\
 \cZ_n \cZ_p &\subset& \cZ_{n+p} \quad\mbox{ for all} \; n,p \,.
\end{eqnarray*}

Evidently the cycle function $a$ determines a homomorphism, $\tilde a$,
\[ \tilde a \,:\,\cZ \longrightarrow U(1) \]
with $\tilde a(\kappa) =a(k_1, \dots,k_n)$ if $\kappa = (k_1,\dots, k_n)^\sim$.
It is clear from this that $\tilde a$ is known on $\cZ_q$ by its values on
the sets $\cZ_n$ for $n < q$ if
\[\cZ_q = \bigcup_{n+p =q, 0<n<q} \cZ_n\, \cZ_p \, .\]

\subsection{Main results}
\begin{theorem}  \label{main}
Let $(X,\mu)$ be a pure point uniformly discrete ergodic point process  and suppose that
\[\underbrace{\Bragg +  \dots + \Bragg}_{n} =\ev  \,.\]
Then  $\cZ$ is generated, as a group, by  $\cZ_{2n+1} $.
\end{theorem}
\noindent
{\sc \bf Proof}: Let ${\bf k} = (k_1, \dots, k_N) $ be any cycle where
$N > 2n+1$. By assumption $k_1 + \dots +k_{n+1} \in \ev$ can be
written in the form $ l_1 + \dots + l_n $, where the $l_i \in \Bragg$. Then
with ${\bf j} := (k_1 , \dots , k_{n+1}) $ and ${\bf l} = (l_1 , \dots , l_n)$
\[ {\bf k} \sim \left(({\bf j})( {\bf -l})\right)\left( {\bf l}(({\bf -j}){\bf k})\right)\]
which writes it as the product of
$({\bf j})( {\bf -l})$ and ${\bf l}(({\bf -j}){\bf k})$. These have reduced lengths
at most $2n+1$ and $N-1$ respectively, and this shows that
$\cZ_{N} \subset \cZ_{2n+1}\cZ_{N-1}$.
The proof finishes by induction. \qed

\begin{theorem}  \label{main2} Let $(X,\mu)$ be a pure point uniformly discrete ergodic point process. If the dynamical spectrum is finitely generated by the diffraction spectrum, then
$\mu$ is uniquely determined by a finite number of its moments $\mu_m$. More precisely, if
\[\underbrace{\Bragg +  \dots + \Bragg}_{n} =\ev  \]
then $\mu$ is uniquely determined by  its moments $\mu_m$, $m=2,\ldots, 2n +1$.
\end{theorem}
\begin{proof}
By (ii) of Thm.~ \ref{momentsToCorrelations}   the second moment determines the autocorrelation
$\gamma$. By Proposition \ref{moments2a} the moments  $\mu_m$, $m=2,\ldots,  2n +1$, then
determine the function $a$ on $\cZ_m$ for $m=2,\ldots, 2n +1$. The
previous theorem gives that the function $a$ is then completely determined. Thus,  by
Proposition \ref{a2moments}, all moments $\mu_m$, $m\geq 2$,  are then determined. Now, the theorem
follows  by (i) of Thm.~ \ref{momentsToCorrelations}.
\end{proof}

\begin{coro} If there are no extinctions, so $\Bragg =\ev$, then
$\mu$ is determined by its second and third correlations.
\end{coro}
\begin{proof} This follows immediately from the previous theorem and parts (ii) and (iii) of Thm.~ \ref{momentsToCorrelations}.
\end{proof}

\section{Model sets}
The theory of model sets is a good place to find examples of the types of point processes
$(X,\mu)$ that we have been discussing.
We start with a cut and project scheme $(\RR^d, H, \tilde{L})$ with corresponding
`torus' $\TT := (\RR^d\times H)/ \tilde{L}$. We suppose that $H$ is a complete metric space. The canonical mapping from the projected image $L$ of $\tilde{L}$
in $\RR^d$ to its projected image in $H$ is denoted by $(\cdot)^\star$. We suppose also that we have a regular model set $\vL = \vL(W) :=\{x \in L \,:\, x^\star \in W \}$ given by some regular closed set $W \subset H$
(i.e. $W = \overline{W^\circ}$ and $W$ has boundary of measure $0$). In this case, the hull $X= X(\vL)$ of $\vL$ is
uniquely ergodic, all elements of $X$ share a common autocorrelation $\gamma$,
and the diffraction $\dif$ is pure point \cite{Martin}.
See \cite{M,BLM} for basic material on model sets.

We wish to consider the situation regarding the diffraction and extinctions.
Let $(\widehat{\RR^d}, \widehat H, \tilde{L}^\circ)$ be the corresponding dual
cut and project scheme, where $\tilde{L}^\circ$ is the dual group of $\TT$,
\cite{M},~Sec.~5. The natural projected image of $\tilde{L}^\circ$ in $\widehat{\RR^d}
\simeq \RR^d$ is denoted by $\ev$. We shall also use the $\star$-notation for the dual
cut and project scheme. Since the projection of
$\tilde{L}^\circ$ into $\widehat H$ has dense image, $\ev^\star$ is dense in
$\widehat H$.

 The diffraction of $\vL$ is known \cite{Hof2,Martin} to be given by
\begin{equation}\label{modelDiffraction}
 \dif(k) = |\widehat{1_W}(-k^\star)|^2  \, ,
 \end{equation}
 for all $k \in \ev$. Here, the measure of the torus $\TT$ is normalized to be one. 

 The set of extinctions $\ext$ is then the set of $k \in \ev$ for which $\widehat{1_W}(-k^\star)$ vanishes, and the Bragg peaks make up the set $\Bragg = \ev \backslash \ext$. We know that $\Bragg$ generates
 $\ev$ as a group and $0 \in \Bragg$.

\begin{lemma}  Suppose that $\overline{\ext^\star}$ has no interior.
Then $\Bragg + \Bragg =\ev$.
\end{lemma}
\begin{proof} We need to prove that every element of $\ext$ is the sum
of two elements of $\Bragg$. Let $K$ be the kernel of $(\phantom{n})^\star$.
From the form of the diffraction (\ref{modelDiffraction}), both $\Bragg$ and $\ext$
consist of unions of full cosets of $K$. Thus it suffices to show that every element of $\ext^\star$ is the sum of two elements of $\Bragg^\star$. Let $z\in \ext^\star$ and suppose that
$z$ is not so expressible, i.e.,
$(z-\Bragg^\star) \cap \Bragg^\star = \emptyset$. Then $z-\Bragg^\star \subset \ext^\star$, so
$\overline{\Bragg^\star} \subset z - \overline{\ext^\star}$. From this
\[ \widehat{H} = \overline{\ev^\star} = \overline{\Bragg^\star \cup \ext^\star} = \overline{\Bragg^\star}
\cup \overline{\ext^\star} \subset (z-\overline{\ext^\star}) \cup \overline{\ext^\star} \, ,\]
which is impossible by the Baire category theorem (since $\overline{\ext^\star}$ has no interior).
\end{proof}

\begin{remark}
One should note that even in the same cut and project scheme, different windows
can give rise to the same diffraction in \eqref{modelDiffraction}, even though the windows are not translationally equivalent. An example of this, that derives from the covariogram problem, is discussed, along with references to covariogram literature,  in \cite{BG}. This shows that the second moment alone is insufficient even to distinguish model sets from the same cut and project scheme.
\end{remark}

\section{Origins of the problem in questions of symmetry}\label{Origins}

The results in this paper have a number of points of contact with the ideas of D. Mermin
\cite{Mermin} and subsequent works of R.~Lifshitz,
D~.A.~Rabson, and B.~N. Fisher \cite{Lif, RF}. The starting point was a puzzle
which arose almost immediately
after the discovery of quasicrystals. The usual notions of symmetry from crystallography are not adequate in the theory of quasicrystals. First of all, translational symmetry is drastically diminished, often to the point of being non-existent. Second, even the finite symmetries are somewhat nebulous. `Icosahedrally symmetric' quasicrystals, for example, have perfectly isosahedrally symmetric diffraction patterns, but they need not be literally icosahedrally symmetric in the sense that the structure is mapped precisely onto itself by the point symmetries of the icosahedral group. In this section we will discuss this question and how it can be resolved through the use of dynamical systems.

Mermin's solution to the symmetry problem was to realize it on the Fourier side of the picture rather than the physical side. Briefly, the idea was to express the density distribution $\rho$ of the quasicrystal
as a superposition of plane waves from its translation module (what we have called
$\ev$ above):
\begin{equation}\label{formalFT}
 \rho(r) = \sum_{k\in \ev} \widehat{\rho}(k) e^{2 \pi i { k \cdot r}} \,,
 \end{equation}
and then to remark that two densities $\rho$ and $\rho'$ based on the
same module of wave vectors are physically
indistinguishable if their correlations of all orders are identical, something that should
happen if
\begin{equation} \label{formalCorrelation}
\widehat{\rho}({k_1}) \dots \widehat{\rho}({k_n}) = \widehat{\rho'}({k_1}) \dots
\widehat{\rho'}({k_n})
\end{equation}
for all $k_1, \dots k_n \in \ev$ with $k_1+ \dots + k_n=0$.  Thus symmetry
becomes symmetry in the sense of indistinguishable correlations.

Now  $\widehat\rho$ and $\widehat{\rho'}$ will be indistinguishable if for some
$\chi:\ev \longrightarrow \RR/\ZZ$ we have
\begin{equation}\label{indist}
\widehat{\rho'}({k}) = e^{2 \pi i \chi(k)} \widehat{\rho}({k}) \, .
\end{equation}
and a symmetry $g$ of $\ev$ would appear as a symmetry of the physical system
if, for all ${k} \in \ev$,
\[ \widehat\rho(g({k})) = e^{2 \pi i \chi_g(k)} \widehat\rho({k}) \]
for some suitable $\chi_g :\ev \rightarrow \RR/\ZZ$.

Applying \eqref{formalCorrelation} at $n=3$ in \eqref{indist} already gives
\[ \chi({k_1} + {k_2})  = \chi({k_1}) +\chi({k_2}) \]
for all ${k_1},{k_2} \in \ev$,
which already provides all the information derivable from \eqref{formalCorrelation}
for all other $n$. It was from this that Mermin concluded that  $2$- and $3$-point correlations should determine everything.

The main difficulty with this approach, and this is already made clear  in \cite{Mermin},
is to give any mathematical meaning to the expressions \eqref{formalFT}, except as tempered distributions. This issue is discussed in some detail by A.~ Hof in \cite{Hof2}. The measure $\rho$ representing the distribution of density of the quasicrystal, say $\rho = \sum_{x\in \vL} \delta_x$, is not in general Fourier transformable as a measure. If instead one treats it as a distribution then it is difficult to say what its Fourier transform looks like and particularly whether or not it is composed of a countable sum of weighted deltas at some (usually dense) subset of $\RR^d$. Certainly the formation of moments in this language would be a formidable task. The upshot of Hof's study of diffraction was his approach to it using the Fourier transform of the autocorrelation (which is Fourier transformable as a measure), and this has been the basis of most subsequent mathematical work on diffraction. Thus \cite{Mermin} is better seen as a formal vision of how things should work out rather than a rigorous exposition of how they actually do, and in that sense it prescient.

In our approach the sums \eqref{formalFT} appear in a form that is tamed by
averaging, namely the expressions appearing in  Thm.~\ref{basic}\,(iii). Such an expression is non-zero
only if ${k} \in \ev$ is the position of a Bragg peak, i.e. there is no extinction there.
In \cite{Mermin} it is claimed that extinctions can only occur due to symmetries, but the examples from model sets show that extinctions come from the Fourier transform
of the window function and do not seem to be related specifically to symmetries.
Furthermore the extinctions seem to be potential, but not absolute, obstructions to the correctness of the assertion about $2$- and $3$-point correlations, see
\S\ref{Final} below. 

In hindsight one can see that the symmetry question amounts to symmetries
of the correlations, hence of the moments, and ultimately of the measure $\mu$ itself.
An isometry $g$ of $\RR^d$ gives rise to a mapping $\vL \mapsto g\vL$. This $g$
is a {\bf symmetry of the point process} if $\mu$ is $g$-invariant. If $X$ is the support of
$\mu$, then this also entails that $g(X) \subset X$.
This is the idea of symmetry put forward by Radin in \cite{Radin0,Radin}.  

This symmetry can also be expressed  by means of  groupoids. In fact, there is a canonical groupoid  structure  associated to  the dynamical system $(X,\RR^d)$. Associated with it is the groupoid of the transversal (see \cite{Bell} and references therein). The point groups of symmetries of the system then act as isomorphisms of this transversal groupoid. The measure $\mu$ appears by giving a trace on the corresponding $C^*$-algebra. In this picture, invariance of the system under a symmetry means invariance of the trace under this symmetry. This invariance of the trace then amounts to invariance of the measure $\mu$, which in turn implies  invariance of diffraction and higher moments under the symmetry.

\section{Final comments}\label{Final}
There are a number of obvious questions that arise from this work. Foremost is the question of whether every cycle on $Z$ arises from some sort of dynamical system of density distributions on $\RR^d$. In \cite{LenzM} we have taken a first look at this. This requires changing the setting from the study of uniformly discrete subsets of $\RR^d$ to more general distributions of density arising by formalizing the main ideas of this
current paper.

There is also the interesting question of the real role of extinctions. Although we have seen that extinctions are an obstruction to our method of reconstructing the measure
$\mu$ from its moments, we have not shown that this obstruction is one of principle rather than an artifact of our approach. Inspired by the present paper this topic has been taken up recently: 
In \cite{XR2} it is shown that for model sets with real internal spaces, whether or not there are extinctions, the $2$- and
$3$-point correlations determine the model set within the class of all model sets. On the other hand, as soon as the internal space is not purely real model sets exist in which the extinctions force higher correlations to be used. The paper \cite{XR3} offers  examples of multi-atomic model sets in which the
scattering strengths of the different point types are differently weighted and for which
even the $2,3,4,5$-point correlations are insufficient to resolve them.

Another question is the curious role of odd numbers in our main result. Are there general conditions under which even numbered higher correlations play the defining role?

Finally it would be interesting to develop practical methods for actually constructing increasingly accurate atomic approximations to the densities using the increasing information from the $2$-, $3$-,  \dots correlations.

\bigskip

\medskip

\textbf{Acknowledgments.} RVM gratefully acknowledges
the support of this research by the Natural Sciences and Engineering Research
Council of Canada. Part of the work was done while DL had a visiting position at Rice University. He would  like to thank the Department of Mathematics there for its hospitality.

\end{document}